\documentclass[authoryear,preprint,review,12pt]{elsarticle}

\usepackage{graphicx}

\usepackage{amssymb}
\usepackage{amsthm}
\usepackage{amsmath}

\biboptions{longnamesfirst,comma}

\journal{Journal of Geometry and Physics}

\usepackage{bm}
\newcommand{\me}{\mathrm{e}}
\newcommand{\mi}{\mathrm{i}}

\newcommand{\dif}{\mathrm{d}}
\DeclareMathAlphabet{\mathsfsl}{OT1}{cmss}{m}{sl}
\newcommand{\tensor}[1]{\mathsfsl{#1}}

\newtheorem{definition}{Definition}[section]
\newtheorem{theorem}{Theorem}[section]

\newtheorem{proposition}{Proposition}[section]

\begin{document}
\begin{frontmatter}

\title{Infinite-dimensional Hamiltonian description of a class of dissipative mechanical systems}

\author[add1]{Tianshu Luo\corref{cor1}}
\ead{ltsmechanic@zju.edu.cn}

\author[add1]{Yimu Guo}
\ead{guoyimu@zju.edu.cn}

\cortext[cor1]{Corresponding author}

\address[add1]{Institute of Applied Mechanics, Department of Mechanics, Zhejiang University, Hangzhou, Zhejiang, 310027,  P.R.China}

\begin{abstract}

In this paper an approach is proposed to represent a class of dissipative mechanical systems by corresponding infinite-dimensional Hamiltonian systems.
This approach is based upon the following structure: for any non-conservative classical mechanical system and arbitrary initial conditions, there 
exists a conservative system; both systems share one and only one common phase curve; and, the value of the Hamiltonian of the conservative 
system is, up to an additive constant, equal to the total energy of the non-conservative system on the aforementioned phase curve, the constant depending 
on the initial conditions. We describe in detail this relationship calling the conservative system ``substitute`` conservative system. 
By considering the dissipative mechanical system as a special fluid in a domain $D$ of the phase space, viz. a collection of particles in this domain, 
we are prompted to develop this system as an infinite-dimensional Hamiltonian system of an ideal fluid.  By comparing the description of the 
ideal fluid in Lagrangian coordinates, we can consider the Hamiltonian and the Lagrangian as the respective integrals of the Hamiltonian and 
the Lagrangian of the substitute conservative system over the initial value space and define a new Poisson bracket to express the equations of motion in
Hamiltonian form. The advantage of the approach is that the value of the canonical momentum density $\pi$ is identical with
 that of the mechanical momentum $m\dot{q}$ and the value of canonical coordinate $q$ is identical with that of the coordinate of the dissipative
mechanical system. Therefore we need not to decouple the Newtonian equations of motion into several one-dimensional ordinary differential equations.
\end{abstract}
\end{frontmatter}


\maketitle

\section{Introduction}

Since the moment Hamilton developed his formalism and equations of motion for classical mechanics, most textbooks
have focused on applying the Hamiltonian formalism in solving solely conservative problems.

\cite{PhysRev.38.815}\cite{book3} gave an example of an artificial Hamiltonian for a damped oscillator based on 
a ``mirror-image'' trick, incorporating a second oscillator with negative friction. The damped oscillator can be represented as 
following:
\begin{equation}
 \ddot{x}+2\eta \dot{x}+\omega^2 x^2=0, \ \ \eta>0.
\label{eq:1d_d_oscl}
\end{equation}
The resulting Hamiltonian is unphysical in the sense that it is unbounded from below (the ground state has infinite negative energy)
 and under time reversal the oscillator is transformed into its ``mirror-image''. 
By this arbitrary trick, dissipative systems can be handled as though they were conservative.  
For the system (\ref{eq:1d_d_oscl}), we have
\begin{eqnarray}
 \ddot{x}+2\eta \dot{x}+\omega^2 x^2&=&0 \ \ (original) \\
 \ddot{y}-2\eta \dot{y}+\omega^2 y^2&=&0 \ \ (mirror-image).
\label{eq:Bateman}
\end{eqnarray}
Associated with these equations of motion is the Bateman(-Morse-Feshbach) Lagrangian
\begin{equation}
 L_B(x,\dot{x},y,\dot{y})=\dot{x}\dot{y}+\eta(x\dot{y}-\dot{x}y)-m\omega^2xy,
\label{eq:L_batm}
\end{equation}
and the two momenta
\[
 p_1=\dot{y}-\eta y,\ \ \ \ p_2=\dot{x}-\eta x.
\]
Therefore a Hamiltonian is defined as
\[
 H=p_1\dot{x}+p_2\dot{y}-L=\dot{x}\dot{y}+\omega^2 xy=(p_1+\eta y)(p_2-\eta x)+\omega^2 xy
\]
Since $y$ increase in amplitude as fast as $x$ decreases, then $H$ will stay constant.

Over six decades ago, \cite{PCalirola1941}\cite{1948PThPh...3..440K} adopted the Hamiltonian
\begin{equation}
 H_{ck}(q,p)=\frac{1}{2}\left(e^{-2\eta t}p^2+e^{2\eta t}\omega^2 q^2\right),
\label{eq:Hck}
\end{equation}
which leads exactly to Eq. (\ref{eq:1d_d_oscl}).
With this Hamiltonian, the canonical momentum is defined as
\[
 p_{ck}=e^{2\eta t}p
\]
It is rather remarkable that \cite{1986AmJPh..54..273H} utilized the same variational principle implicit in the above, but did not 
change or redefine the Lagrangian and momentum. The action was defined as
\[
 I=\int_{t_1}^{t_2} \me^{2\eta t}L\dif t.
\]
Not changing the definition of momentum may be potentially a good characteristic.

In the 1960s, the paper of \cite{art3} utilized the classical Hamiltonian formalism and a perturbation theory to solve a non-conservative problem. They did 
not attempt to derive the Hamiltonian formalism for non-conservative problems. Several authors have attempted to expand the scope of 
Hamiltonian formalism to include dissipative problems. 

Some significant works in this area were done by \cite{Vujanovic1970,Vujanovic1978} 
and \cite{Djukic1973,Djukic1975,Djukic1975a}.  They have proposed a technique for systems with gauge variant Lagrangian. 
 \cite{Vujanovic1970,Vujanovic1978} defined a universal Lagrangian
\[
\bar{L}\left(\bar{x},\frac{\dif \bar{x}}{\dif \bar{t}},\bar{t}\right)=L(x,\dot{t},t)+\frac{\dif}{\dif t}f(t,x,\dot{x}),
\]
where the function $f$ is known as the gauge function, $L$ satisfied the equations
\[
 \frac{\dif}{\dif t}\frac{\partial L}{\partial \dot{x}}-\frac{\partial L}{\partial x}=F(t,x,\dot{x}).
\]
In the equation above, $F$ is a nonconservative force that is not derivable from a potential function. Their technique might 
include the aforementioned cases. \cite{Mukherjee2006a} had deemed the technique as being rather algebraic in nature. To overcome the limitations,
 Earlier \cite{Mukherjee1994} had proposed a modified equation that overcomes that limitation by introducing an additional time-like variable 
called `umbra time' and extended this notion to the co-kinetic  kinetic, potential, complimentary energies as well as the Lagrangian itself.
 Following up on this work, \cite{Mukherjee1997} introduced a procedure by which to obtain the umbra-Lagrangian via system bond graphs, thus 
extending the basic idea of \cite{Karnopp1977}. \cite{Mukherjee2001} consolidated this idea and presented an important notion of 
invariants of motion. A gauge variant Lagrangian implies a new definition for canonical momentum, which might not be identical 
with the mechanical momentum.

\cite{McLachlan2001276} proposed conformal Hamiltonian notion to describe damping mechanical system. In conformal Hamiltonian 
description, a new Poisson bracket was defined. Vector fields whose flow preserves a symplectic form up to a constant, such as 
simple mechanical systems with friction, are called “conformal”. They developed a reduction theory for symmetric conformal 
Hamiltonian systems, analogous to symplectic reduction theory.

\cite{art7} considered that a large class of dissipative systems that can be brought to a canonical form by introducing complex
 coordinates in phase space and a complex-valued Hamiltonian. He showed indicated that Eq. (\ref{eq:1d_d_oscl}) can be brought to diagonal
 form by a linear transformation:
\begin{equation}
 z=A\left[-\mi(p+\eta x)+\omega_1 x\right], \frac{\dif z}{\dif t}=\left[-\gamma+\mi \omega_1\right]z,
\label{eq:complex_transform}
\end{equation}
where 
\begin{equation}
 \omega_1=\sqrt{\omega_2-\gamma^2},
\end{equation}
and the constant $A=1/\sqrt{2\omega_1}$. Next, the complex-valued function as Hamiltonian function is defined as
\begin{equation}
 \mathcal{H}=(\omega_1+\mi \eta)zz^*,
\end{equation}
which satisfies
\[
 \frac{\dif z}{\dif t}=\left\lbrace \mathcal{H},z \right\rbrace, \ \ \frac{\dif z^*}{\dif t}=\left\lbrace \mathcal{H},z^* \right\rbrace
\]

The previous works of \cite{PhysRev.38.815},\cite{book3},\cite{PCalirola1941},\cite{1948PThPh...3..440K},\cite{Vujanovic1970,Vujanovic1978} 
,\cite{McLachlan2001276} leaded to a finite-dimensional Hamiltonian mechanics with symplectic structure(standard Hamiltonian mechanics). 
Employing these approaches, redefinition of momentum is generally required. Our aim here is to seek an approach without redefinition of momentum by applying the formalism of 
infinite-dimensional Hamiltonian mechanics without symplectic structure to classical 
dissipative mechanical problems. Dissipative mechanical systems, sometimes called Chetayev systems, are usually represented as
\begin{equation}
 \ddot{\bm{q}}+\tensor{C}\dot{\bm{q}}+\tensor{K}\bm{q}=0,
\label{eq:eq1}
\end{equation}
where $\bm{q}$ is a vector formed from the generalized coordinates identified in the system, 
$\tensor{C}$ denotes a non-linear damping coefficient matrix which depends on $\bm{q}$, and $\tensor{K}$ denotes a non-linear stiffness 
matrix also depending on $\bm{q}$ and consists of two parts $\tensor{K}=\tensor{\check{K}}+\tensor{\hat{K}}$($\tensor{\check{K}}$ is a diagonal matrix). 
In light of a proposition put forward by \cite{2010arXiv1007.2709L}, we show how to represent the dissipative mechanical 
system (\ref{eq:eq1}) as an infinite-dimensional Hamiltonian system. This proposition asserts that for any dissipative 
classical mechanical system with arbitrary initial condition, there exists an associated conservative system; both systems 
sharing one and only one common phase curve; the energy of the conservative system is the sum of the total energy of the 
dissipative system on the aforementioned phase curve and a constant depending on the initial condition. In sec. \ref{sec:2} 
a demonstration of the proposition is reported for the first time. In this demonstration, the equations below will be implemented 
\begin{eqnarray}
\dot{p}_i&=&- \frac{\partial H}{\partial q_i}+\bm{F}\left( \frac{\partial{r}}{\partial q_i} \right) \nonumber \\
\dot{q}_i&=& \frac{\partial H}{\partial p_i},
\label{eq:eq2}
\end{eqnarray}
where $\lbrace q_i,p_i \rbrace$ denote the coordinate and momentum, and the position vector $r$ depends on $\lbrace q_i\rbrace$, 
i.e. $r(q_i)$, $H$ denotes the mechanical energy, and $\bm{F} (\partial{r}/ \partial{q_i})$ denotes the $i^{th}$ generalized force which 
depends on $q_i,p_i$. \cite{Marsden2007}  and other researchers applied these equations (\ref{eq:eq2}) to the problem related to stability 
in dissipative systems.\cite{Marsden2007} considered systems described by that Eq. (\ref{eq:eq2}) that were composed of a conservative part 
and a non-conservative part. However, Eq. (\ref{eq:eq2}) apparently is not a set pf Hamilton's equations of motion, but can be 
considered as a representation of the dissipative mechanical system (\ref{eq:eq1}) in phase space. Indeed the system (\ref{eq:eq2}) 
is more general than the system (\ref{eq:eq1}), and the system (\ref{eq:eq2}) is equivalent to the system of differential 
equations considered by \cite{JesseDouglas1941}
\[
 \ddot{y_i}=G_i(t,y_j,\dot{y_j}),
\]
which can be reduced to Eq. (\ref{eq:eq1}) by a first-order linearization procedure.
Analogous to the Hamiltonian description of an ideal fluid in Lagrangian variables and that of Poisson-Vlasov equations, 
we define the Lagrangian and Hamiltonian as an integral over the entire initial value space. The generalized coordinates 
and the canonical momentum will be thought of as functions of the initial value and time. A new Poisson bracket will be 
defined to represent Eq. (\ref{eq:eq1}) as an infinite-dimensional 
Hamiltonian system. This process will be presented in detail in Sec. \ref{sec:3}.

\section{Corresponding Conservative Mechanical Systems}\label{sec:2} 
\subsection{Common Phase Flow Curve}\label{sec2.1}
To begin, we need to represent Eq. (\ref{eq:eq1}) in the form Eq. (\ref{eq:eq2}). Under general circumstances, the force $\bm{F}$ is 
assumed to be a non-conservative force that depends on the variable set $q_1,\cdots,q_n,\dot{q}_1,\cdots,\dot{q}_n$
\[
 \bm{F}=-\tensor{C}\dot{\bm{q}}-\tensor{\hat{K}}\bm{q},  
\]
and the mechanical energy $H$ is
\[
 H=\frac{1}{2}\bm{p}^T\bm{p}+\frac{1}{2}\bm{q}^T\tensor{\check{K}}\bm{q},
\]
where superscript $T$ denotes matrix transposition.
We denote by $F_i$ the components of the generalized force $\bm{F}$.
\begin{equation}
F_i(q_1,\cdots,q_n,\dot{q}_1,\cdots,\dot{q}_n)=\bm{F}\left( \frac{\partial{r}}{\partial q_i}\right).
\label{eq:inth-1}
\end{equation}
Thus we can reformulate Eq. (\ref{eq:eq2}) as follows:
\begin{eqnarray}
\dot{p}_i&=& -\frac{\partial H}{\partial q_i}
+F_i(q_1,\cdots,q_n,\dot{q}_1,\cdots,\dot{q}_n) \nonumber \\
 \dot{q}_i&=&\frac{\partial H}{\partial p_i}.
\label{eq:inth-2}
\end{eqnarray}
Suppose that we have the Hamiltonian, denoted by $\hat{H}$ of a conservative system non-damped system. Thus we
 may write down the Hamilton's equation of motion for this conservative system:
\begin{eqnarray}
\dot{p}_i &=& -\frac{\partial {\hat{H}}}{\partial q_i} \nonumber \\
\dot{q}_i &=&\frac{\partial \hat{H}}{\partial p_i}.
\label{eq:inth-3}
\end{eqnarray}
We assume that the standard definition of momentum in Hamiltonian formulation still holds, but we do require that a special solution  
of Eq. (\ref{eq:inth-3}) is the same as that of Eq. (\ref{eq:inth-2}).
We also assume that a phase curve\footnote{This concept appears in the book of \cite{Arnold_ode}:``The orbits of a phase flow called phase curves''} 
$\gamma$ of Eq. (\ref{eq:inth-2}) coincides with 
that of Eq. (\ref{eq:inth-3}). This phase curve $\gamma$ corresponds to an initial condition $q_{i0},p_{i0}$. 
Consequently by comparing Eq. (\ref{eq:inth-2}) and Eq. (\ref{eq:inth-3}), we have
\begin{eqnarray}
\left.\frac{\partial{\hat{H}}}{\partial{q_i}}\right|_{\gamma} &=&
\left.\frac{\partial H}{\partial q_i}\right|_\gamma-\left. F_i(q_1,\cdots,q_n,\dot{q}_1,\cdots,\dot{q}_n)\right|_\gamma \nonumber \\
\left.\frac{\partial{\hat{H}}}{\partial{p_i}}\right|_\gamma&=&
\left.\frac{\partial H}{\partial p_i}\right|_\gamma,
\label{eq:inth-4}
\end{eqnarray}
where $\left.\frac{\partial{\hat{H}}}{\partial{q_i}}\right|_{\gamma},\left.\frac{\partial H}{\partial q_i}\right|_\gamma
,\left.\frac{\partial{\hat{H}}}{\partial{p_i}}\right|_\gamma$ and $\left.\frac{\partial H}{\partial p_i}\right|_\gamma$ denote the values
of these partial derivatives on the phase curve $\gamma$ and similarly $\left.F_i(q_1,\cdots,q_n,\dot{q}_1,\cdots,\dot{q}_n)\right|_\gamma$ denotes 
the value of the force $F_i$ on the phase curve $\gamma$. 
In classical mechanics, the mechanical energy $H$ of the system (\ref{eq:inth-2}) 
can be evaluated via the following equation
\begin{equation}
 H=\int_{\gamma}\left(\frac{\partial{H}}{\partial{q_i}}\right)\dif q_i
+\int_{\gamma} \left(\frac{\partial H}{\partial p_i}\right)\dif p_i+const_1,
\label{eq:inth-5}
\end{equation}
where $const_1$ is a constant that depends on the initial condition described above. 
The Einstein summation convention is being employed here and in the rest of this section. If $q_i=0,p_i=0$, then $const_1=0$.
The Hamiltonian $\hat{H}$ of the conservative mechanical 
system (\ref{eq:inth-3}) is mechanical energy and can be written as:
\begin{equation}
 \hat{H}=\int_{\gamma}\left(\frac{\partial{\hat{H}}}{\partial{q_i}}\right)\dif q_i
+\int_{\gamma} \left(\frac{\partial{\hat{H}}}{\partial p_i}\right)\dif p_i+const_2,
\label{eq:inth-7}
\end{equation}
where $const_2$ is a constant that depends on the initial condition described above. 
The following equations can be derived from Eq. (\ref{eq:inth-4})
\begin{eqnarray}
 \int_{\gamma}\left(\frac{\partial{\hat{H}}}{\partial{q_i}}\right)\dif q_i
&=&\int_{\gamma}\left[\left(\frac{\partial H}{\partial q_i}\right)
-F_i(q_1,\cdots,q_n,\dot{q}_1,\cdots,\dot{q}_n)\right]\dif q_i \nonumber \\
 \int_{\gamma} \left(\frac{\partial \hat{H}}{\partial p_i}\right)\dif p_i
&=&\int_{\gamma} \left(\frac{\partial H}{\partial p_i}\right)\dif p_i.
\label{eq:inth-6}
\end{eqnarray}
 
Substituting Eqs.(\ref{eq:inth-5})and Eqs.(\ref{eq:inth-6}) into Eq. (\ref{eq:inth-7}), we have
\begin{equation}
 \hat{H}=H-\int_{\gamma}F_i(q_1,\cdots,q_n,\dot{q}_1,\cdots,\dot{q}_n)\dif q_i+const.
\label{eq:inth-8}
\end{equation}
where $const=const_2-const_1$.
According to the physical interpretation of this Hamiltonian, $const_1$, $const_2$ and $const$ are added to Eq. (\ref{eq:inth-5})(\ref{eq:inth-7})(\ref{eq:inth-8})
respectively so that the integral constant vanishes in the Hamiltonian quantity. 

\cite{Arnold1997} had presented the Newton-Laplace principle of determinacy as, 
'This principle asserts that the state of a mechanical system at any fixed moment of time uniquely
determines all of its (future and past) motion.' In other words, in phase space the position and the velocity variables are 
functions of only the time $t$. We can therefore assume that we already have a solution of Eq. (\ref{eq:inth-2})
\begin{eqnarray}
 q_i&=&q_i(t) \nonumber \\
 \dot{q_i}&=&\dot{q_i}(t),
\label{eq:curve}
\end{eqnarray}
where the solution satisfies the initial conditions. We can divide the whole time domain into a group of sufficiently small domains and 
in these domains, $q_i$ is monotonic, and hence we can assume an inverse function $t=t(q_i)$. If $t=t(q_i)$ is substituted into the 
non-conservative force $\left.F_i\right|_{\gamma}$, we find that:
\begin{equation}
\left.F_i(q_1(t(q_i)),\cdots,q_n(t(q_i)),\dot{q}_1(t(q_i)),\cdots,\dot{q}_n(t(q_i)))\right|_{\gamma}= \mathcal{F}_i(q_i),
\label{eq:asumption}
\end{equation}
where $\mathcal{F}_i$ is a function of $q_i$ alone. In Eq. (\ref{eq:asumption}) the function $F_i$ is restricted to the curve $\gamma$, such that a new function
 $\mathcal{F}_i(q_i)$ yields. Thus we have
\begin{eqnarray}
 \int_{\gamma}F_i \dif q_i&=&\int_{q_{i0}}^{q_i}\mathcal{F}_i(q_i) \dif q_i
=W_i(q_i)-W_i(q_{i0}).
\label{eq:inth-9}
\end{eqnarray}
According to Eq. (\ref{eq:inth-9}) the function $\mathcal{F}_i$ is path independent, and therefore $\mathcal{F}_i$ can be regarded as a conservative force. 
From this, Eq. (\ref{eq:asumption}) represents an identity map from the non-conservative force $F$ on the curve $\gamma$ 
to the conservative force $\mathcal{F}_i$ which is distinct from $F_i$. Eq. (\ref{eq:asumption}) is tenable only on the phase curve $\gamma$.
 Consequently the function form of $\mathcal{F}_i$ depends on the aforementioned initial condition; from other initial conditions, $\mathcal{F}_i$
with different function forms will be obtained. 

Again, according to the physical meaning of the Hamiltonian, a $const$ is added to Eq. (\ref{eq:inth-8}) so
that the integral constant vanishes in the Hamiltonian quantity. Hence setting $const=-W_i(q_{i0})$ then
substituting this choice and Eq. (\ref{eq:inth-9}) into Eq. (\ref{eq:inth-8}), we have
\begin{equation}
\hat{H}=H-W_i(q_i)
\label{eq:inth-10}
\end{equation}
where $-W_i(q_i)$, a function of $q_i$ denotes the potential of the conservative force $\mathcal{F}_i$; $W_i(q_i)$ is then equal to the sum of the work done by the 
non-conservative force $F$ and $const$. In Eq. (\ref{eq:inth-10}) $\hat{H}$ and $H$ are both functions of $q_i$ and $p_i$.
Eq. (\ref{eq:inth-10}) and Eq. (\ref{eq:inth-8}) can be thought of as a map from the total energy of the dissipative system (\ref{eq:inth-2}) involving 
the mechanical energy and the energy lost to the Hamiltonian of the conservative system (\ref{eq:inth-3}). Indeed, $\hat{H}$ and the total energy differ by 
the constant $const=-W_i(q_{i0})$. 

Based on the above, the following proposition can be stated:
\begin{proposition}
For any non-conservative classical mechanical system and arbitrary initial condition, there exists a conservative system; both systems sharing 
one and only one common phase curve; and the value of the Hamiltonian of the conservative system is equal to the sum of the total energy 
of the non-conservative system on the aforementioned phase curve and a constant depending on the initial condition.
\label{pro:1}
\end{proposition}
\begin{proof} 

To begin, we must prove the first part of the Proposition \ref{pro:1}, i.e. that a conservative system with Hamiltonian presented by Eq. (\ref{eq:inth-10.1})
shares a common phase curve with the non-conservative system represented by Eq. (\ref{eq:inth-2}). In other words,  
the Hamiltonian quantity presented by Eq. (\ref{eq:inth-10.1}) satisfies Eq. (\ref{eq:inth-4}) under the same initial condition. 
Substituting Eq. (\ref{eq:inth-10.1}) into the left side of Eq. (\ref{eq:inth-4}), we have
\begin{eqnarray}
\frac{\partial{\hat{H}(q_i,p_i)}}{\partial {q_i}}&=&\frac{\partial H(q_i,p_i)}{\partial {q_i}}
-\frac{\partial{W_j(q_j)}}{\partial {q_i}} \nonumber\\
\frac{\partial{\hat{H}(q_i,p_i)}}{\partial {p_i}}&=&\frac{\partial H(q_i,p_i)}{\partial {p_i}}
-\frac{\partial{W_j(q_j)}}{\partial {p_i}}.
\label{eq:inth-11}
\end{eqnarray}
It must be noted that, although $q_i$ and $p_i$ are considered as distinct variables in Hamiltonian mechanics, we can consider $q_i$ and
$p_i$ as dependent variables in constructing of $\hat{H}$.
On the phase curve $\gamma$ we have
\begin{eqnarray}
 \frac{\partial{{W_j(q_j) }}}{\partial {q_i}}&=&
\frac{\partial{(\int_{q_{j0}}^{q_j}\mathcal{F}_j(q_j) \dif q_j+W_i(q_{i0}))}}{\partial {q_i}}
=\mathcal{F}_i(q_i) \nonumber \\
\frac{\partial{{W_j(q_j) }}}{\partial {p_i}}&=0,
\label{eq:inth-12}
\end{eqnarray}
where $\mathcal{F}_i(q_i)$ is equal to the damping force $F_i$ on $\gamma$. Hence under the initial condition $q_0, p_0$, 
Eq. (\ref{eq:inth-4}) is satisfied. As a result, we can state that the phase curve  of Eq. (\ref{eq:inth-3}) coincides with that of 
Eq. (\ref{eq:inth-2}) subject to the initial condition;  and $\hat{H}$ represented by Eq. (\ref{eq:inth-10.1}) is 
the Hamiltonian of the conservative system represented by Eq. (\ref{eq:inth-3}).

Next, we must prove the second part of Proposition \ref{pro:1}: the uniqueness of the common phase curve.

 We assume that Eq. (\ref{eq:inth-3}) shares two common phase curves, $\gamma_1$ and $\gamma_2$, with Eq. (\ref{eq:inth-2}). 
Let $z_1$ be a point of $\gamma_1$ at the time $t$, $z_2$ a point of $\gamma_2$ at the time $t$, and 
$g^t$ the Hamiltonian phase flow of Eq. (\ref{eq:inth-3}). Suppose that a domain $\Omega$ at $t$ which contains only points $z_1$ and $z_2$, is 
not only a subset of the phase space of the non-conservative system (\ref{eq:inth-2}) but also that of the phase space of the conservative system (\ref{eq:inth-3}). 
Hence, there exists a phase flow $\hat{g}^t$ composed of $\gamma_1$ and $\gamma_2$, that is the phase flow of Eq. (\ref{eq:inth-2}) restricted by $\Omega$.
 According to the following form of Liouville's theorem as stated in \cite{Arnold1978}:
\begin{theorem}
The phase flow of Hamilton's equations of motion preserves volume: for any region $D$ in the phase space we have
\[
 volume\ of\ g^tD=volume\ of\ D
\]
where $g^t$ is the one-parameter group of transformations of phase space
\[
 g^t:(p(0),q(0))\longmapsto:(p(t),q(t)),
\]
\label{Liouville}
\end{theorem}
that preserves the volume of $\Omega$. This implies that the phase flow of Eq. (\ref{eq:inth-2}) $\hat{g}^t$ also preserves the volume of $\Omega$. 
 However, system (\ref{eq:inth-2}) is not conservative, which conflicts
with Louisville's theorem; hence only a phase curve of Eq. (\ref{eq:inth-3}) coincides with that of Eq. (\ref{eq:inth-2}).

\qed
\end{proof}
In the next subsection a simple example is given to illustrate  the process involved in applying Proposition \ref{pro:1}. 
\subsection{One-dimensional Simple Example}
Consider a special one-dimensional simple mechanical system:
\begin{equation}
 \ddot{x}+c\dot{x}=0,
\label{eq:simp_1d}
\end{equation}
where $c$ is a constant. The exact solution of the equation above is
\begin{equation}
 x=A_1+A_2e^{-ct},
\label{eq:sol_1d}
\end{equation}
where $A_1,A_2$ are constants. Differentiation gives the velocity:
\begin{equation}
 \dot{x}=-cA_2e^{-ct}.
\label{eq:sol_1dv}
\end{equation}
From the initial condition $x_0,\dot{x}_0$, we find $A_1=x_0+\dot{x}_0/c, A_2=-\dot{x}_0/c$. Inverting Eq. (\ref{eq:sol_1d}) yields
\begin{equation}
 t=-\frac{1}{c}\ln\frac{x-A_1}{A_2}
\label{eq:tfunc}
\end{equation}
and by substituting into Eq. (\ref{eq:sol_1dv}), such we have
\begin{equation}
 \dot{x}=-c(x-A_1)
\label{eq:dx}
\end{equation}
The dissipative force $F$ in the dissipative system (\ref{eq:simp_1d}) is
\begin{equation}
 F=c\dot{x}.
\label{eq:F}
\end{equation}
Substituting Eq. (\ref{eq:dx}) into Eq. (\ref{eq:F}), the conservative force $\mathcal{F}$ is expressed as
\begin{equation}
 \mathcal{F}=-c^2(x-A_1);
\label{eq:mF}
\end{equation}
Clearly, the conservative force $\mathcal{F}$ depends on the initial condition of the dissipative system (\ref{eq:simp_1d}), in other words,  
an initial condition determines a conservative force. Consequently, a new conservative system yields
\begin{equation}
 \ddot{x}+\mathcal{F}=0\rightarrow \ddot{x}-c^2(x-A_1)=0.
\label{eq:1d_eq_consys}
\end{equation}
The stiffness coefficient in this equation must be negative. One can readily verify that the particular solution (\ref{eq:sol_1d}) 
of the dissipative system can satisfy the conservative one (\ref{eq:1d_eq_consys}). This point agrees with Proposition (\ref{pro:1}).

The potential of the conservative system (\ref{eq:1d_eq_consys}) is 
\[
 V=\int_0^x \left[ -c^2(x-A_1) \right]\dif x =-\frac{c^2}{2}x^2+c^2A_1x 
\]
Therefore the Hamiltonian is
\[
 \hat{H}=T+V=\frac{1}{2}p^2-\frac{c^2}{2}x^2+c^2A_1x,
\]
where $p=\dot{x}$.

Furthermore, Proposition (\ref{pro:1}) can be depicted by Fig. \ref{fig:relation}. The phase flow of conservative system (\ref{eq:sol_1d}) 
transforms the red area in phase space to the purple area; the phase flow of conservative system (\ref{eq:1d_eq_consys}) transforms 
the red area to the green area. The blue curve in Fig. \ref{fig:relation} illustrates the common phase curve. If one draws more common phase 
curves and phase flows, the picture will like a flower, the phase flow of the nonconservative system likes a pistil and phase flows conservative systems like 
petals.

\begin{figure}
\begin{center}
\scalebox{0.4}{\includegraphics{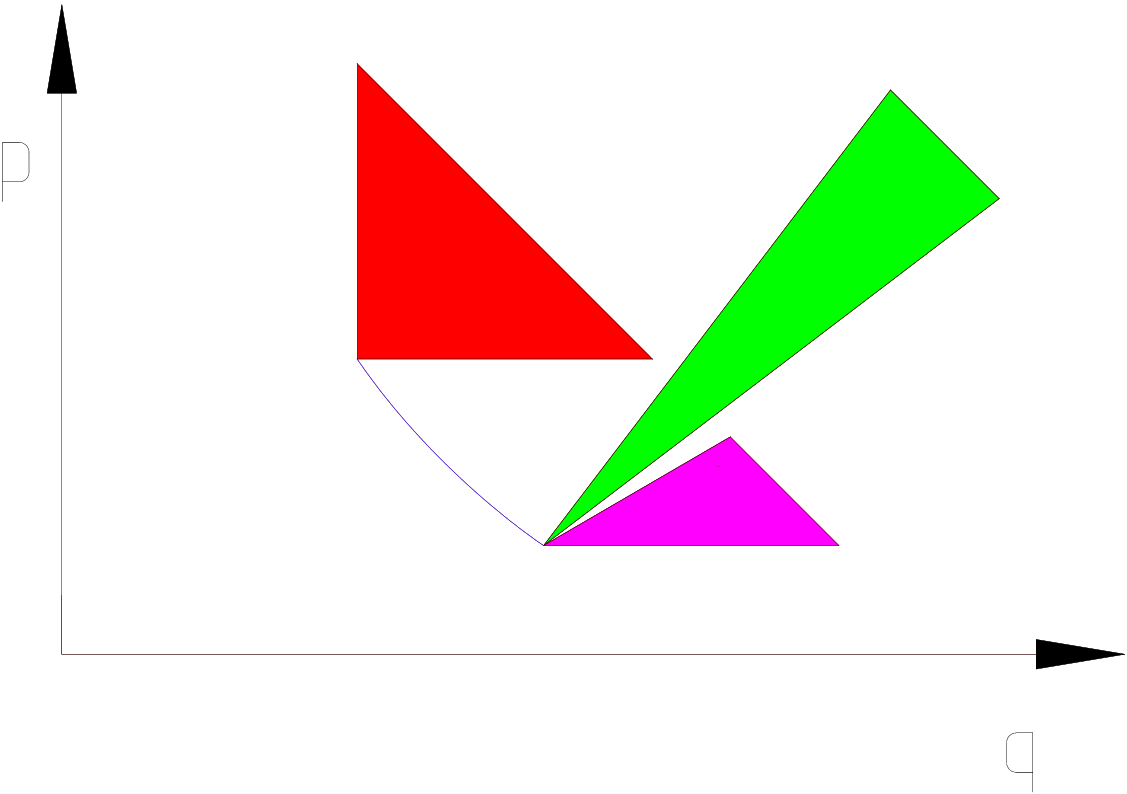}}
\caption{Relationship between nonconservative system (\ref{eq:simp_1d}) and conservative one (\ref{eq:1d_eq_consys})} 
\end{center}
\label{fig:relation}       
\end{figure}

\subsection{Obtaining the Equivalent Stiffness Matrix $\tilde{K}$} \label{sec:2.2}
In accordance with Proposition \ref{pro:1}, a conservative mechanical system was found associated with the dissipative system (\ref{eq:eq1}) in 
addition to its initial conditions. Subject to these initial conditions, the dissipative system (\ref{eq:eq1}) possesses a common phase curve $\gamma$
 with the conservative system. As in Eq. (\ref{eq:asumption}), we can consider that the components of the damping force $\tensor{C}\dot{\bm{q}}$ determine 
the components of a conservative force on the phase curve $\gamma$
\begin{equation}
\begin{array}{ccc}
c_{11}\dot{q}_1=\varrho_{11}(q_1)&\dots&c_{1n}\dot{q}_n=\varrho_{1n}(q_1)\\
\vdots&\ddots&\vdots\\
c_{n1}\dot{q}_1=\varrho_{21}(q_n)&\dots&c_{nn}\dot{q}_n=\varrho_{nn}(q_n).
\end{array}
\label{eq:ex2-4}
\end{equation}
For convenience, this conservative force is assumed to be an elastic restoring force: 
\begin{equation}
\begin{array}{ccc}
\varrho_{11}(q_1)=\kappa_{11}(q_1)q_1&\dots&\varrho_{1n}(q_1)=\kappa_{1n}(q_1)q_1\\
\vdots&\ddots&\vdots\\
\varrho_{n1}(q_1)=\kappa_{n1}(q_n)q_n&\dots &\varrho_{nn}(q_n)=\kappa_{nn}(q_n)q_n .
\end{array}
\label{eq:ex2-5}
\end{equation}

In a similar manner, the components of the non-conservative force $\tensor{\hat{K}}\bm{q}$ are equal to the 
components of a conservative force on the phase curve $\gamma$
\begin{equation}
\begin{array}{ccc}
\hat{K}_{11}q_1=\chi_{11}(q_1)&\dots&\hat{K}_{1n}q_n=\chi_{1n}(q_1)\\
\vdots&\ddots&\vdots\\
\hat{K}_{n1}q_1=\chi_{21}(q_n)&\dots&\hat{K}_{nn}q_n=\chi_{nn}(q_n).
\end{array}
\label{eq:ex2-4a}
\end{equation}
The conservative force can likewise be assumed to an elastic restoring force: 
\begin{equation}
\begin{array}{ccc}
\chi_{11}(q_1)=\lambda_{11}(q_1)q_1&\dots&\chi_{1n}(q_1)=\lambda_{1n}(q_1)q_1\\
\vdots&\ddots&\vdots\\
\chi_{n1}(q_1)=\lambda_{n1}(q_n)q_n&\dots &\chi_{nn}(q_n)=\lambda_{nn}(q_n)q_n .
\end{array}
\label{eq:ex2-5a}
\end{equation}
By an appropriate transformation, an equivalent stiffness matrix $\tensor{\tilde{K}}$ that is diagonal in form can be obtained
\begin{equation}
\tensor{\tilde{K}}_{ii}=\sum_{l=1}^n \kappa_{il}(q_l)+\lambda_{il}(q_l).
\label{eq:ex2-5b}
\end{equation}

Consequently, an $n$-dimensional conservative system is obtained
\begin{equation}
 \bm{\ddot{q}}+(\tensor{\check{K}}+\tensor{\tilde{K}})\bm{q}=0
\label{eq:ex2-6}
\end{equation}
which shares the common phase curve $\gamma$ with the $n$-dimensional damping system described by (\ref{eq:eq1}). 
In this paper, the conservative system is called the 'substitute' conservative system.
The Lagrangian of Eqs.(\ref{eq:ex2-6}) is
\begin{equation}
 \hat{L}=\frac{1}{2}\dot{\bm{q}}^T\dot{\bm{q}}-\int_{\bm{0}}^{\bm{q}}(\tensor{\check{K}}\bm{q})^T\dif \bm{q}-
\int_{\bm{0}}^{\bm{q}} (\tilde{\tensor{K}}\bm{q})^T \dif \bm{q},
\label{eq:ex2-7a}
\end{equation}
with the Hamiltonian
\begin{equation}
 \hat{H}=\frac{1}{2}\bm{p}^T\bm{p}+\int_{\bm{0}}^{\bm{q}}(\tensor{\check{K}}\bm{q})^T\dif \bm{q}+
\int_{\bm{0}}^{\bm{q}} (\tilde{\tensor{K}}\bm{q})^T \dif \bm{q},
\label{eq:ex2-7}
\end{equation}
where $\bm{0}$ is the zero vector, and $\bm{p}=\dot{\bm{q}}$. Here $\hat{H}$ in Eq. (\ref{eq:ex2-7}) is the mechanical energy of the 
conservative system (\ref{eq:ex2-6}), because $\int_{\bm{0}}^{\bm{q}} (\tilde{\tensor{K}}\bm{q})^T \dif \bm{q}$
 is a potential function such that $\hat{H}$ is independent of the path taken in phase space. It should be noted, that this procedure in 
Sec. \ref{sec:2.2} should be carried out via numerical methods.

\section{Derivation of Hamiltonian Description of Dissipative Mechanical Systems}\label{sec:3}
\cite{RickSalmon1988} and \cite{RevModPhys.70.467} described in detail the Hamiltonian formalism of an ideal fluid in Lagrangian 
variables in detail. In the Hamiltonian description, a fluid is described as a collection of fluid 
particles or elements in the domain $D$. The Lagrangian and Hamiltonian of an ideal fluid, which are the sum of the respective Lagrangian 
and Hamiltonian of particles (Lagrangian density and Hamiltonian density) in the domain $D$, are integrals over this domain in an initial configuration 
space. In classical mechanics context, Eq. (\ref{eq:eq1}) can be implemented to describe motion of a 
set of particles in an n-dimensional configuration space. We then ask whether we can regard the sum of Lagrangian (resp. Hamiltonian)
 of these particles as Lagrangian (resp. Hamiltonian) of the whole set.  
In the previous section, we had demonstrated that a phase curve for conservative
system (\ref{eq:ex2-6}) is identical with one for the dissipative system (\ref{eq:eq1}). 
The question then is why not consider the Lagrangian of system (\ref{eq:ex2-6}) as a Lagrangian density?

To begin, Let us review the Hamiltonian description of an ideal fluid 
in Lagrangian variables, and then give the infinite dimensional Hamiltonian formulation of system (\ref{eq:eq1}). 
Suppose the coordinate of a fluid particle at time $t$ is
\begin{equation}
\bm{q}=\bm{q}(\bm{a},t),
\label{eq:fluid_cod}
\end{equation}
where $\bm{q}=\{q_1,q_2,q_3\}$, and $\bm{a}=\{a_1,a_2,a_3\}$ is the position of the particle at initial time $t=t_0$. 
We assume that $\bm{a}$ varies over a fixed domain $D$, which is completely filled with fluid, and that the functions $q$ map $D$ onto itself.

In  Lagrangian variables $\bm{a}$, the Lagrangian density of the fluid particles is 
\begin{equation}
 \mathcal{L}_{f}(\bm{q},\dot{\bm{q}},\partial \bm{q}/\partial \bm{a},t )
=\frac{1}{2}\rho_0\dot{\bm{q}}^2-\rho_{0}E(s_0,\rho_0/\tensor{\mathcal{J}})-\phi,
\label{eq:fluid_lagdens}
\end{equation}
where $\rho_0=\rho_0(\bm{a})$ is a given initial density distribution, $\dot{\bm{q}}$ is the velocity of the fluid particle, a shorthand 
$\dot{\bm{q}}^2=\delta_{ij}q_iq_j$ is used, $E$ is the energy per unit mass, $s_0$ is the entropy per unit mass at time
 $t_0$, $\tensor{\mathcal{J}}=\det(\partial{q}^i/\partial{a}^j)$, 
and $\phi$ is a potential function for external conservative forces. The intensive quantities, pressure and temperature are
obtained as follows:
\begin{equation}
 T=\frac{\partial E}{\partial s}(s,\rho),\ \ p=\rho^2\frac{\partial E}{\partial \rho}(s,\rho)
\end{equation}

Therefore, we have the Lagrangian functional of the fluid particles of the domain $D$:
\begin{equation}
 L_f(\bm{q},\dot{\bm{q}})=\int_D \mathcal{L}_f\dif^{3}\bm{a}
=\int_D \left[\frac{1}{2}\rho_0\dot{\bm{q}}^2-\rho_{0}E(s_0,\rho_0/\tensor{\mathcal{J}})-\phi\right]\dif^{3}\bm{a},
\label{eq:fluid_lag}
\end{equation}
where $\dif^3\bm{a}=\dif a_1\dif a_2 \dif a_3$. 
Thus the action functional is given by
\begin{equation}
 S_f[\bm{q}]=\int_{t_0}^{t^1}\dif t\int_D L_f[\bm{q},\dot{\bm{q}}]\dif^3 \bm{a}
=\int_{t_0}^{t^1}\dif t\int_D \left[\frac{1}{2}\rho_0\dot{\bm{q}}^2-\rho_{0}E-\phi\right]\dif^{3}\bm{a}
\label{eq:fluid_action}
\end{equation}
Observe that this action functional is akin to that for a system of
finite-degree-of-freedom, as treated above, except that the sum over particles is replaced by integration over $D$, i.e.,
\begin{equation}
 \int_D \dif^3 \bm{a}\leftrightarrow \sum_i
\label{eq:int2sum}
\end{equation}

By a functional differentiation, we have a canonical momentum density 
\begin{equation}
 \bm{\varpi}_i(\bm{a},t)=\frac{\delta L_f}{\delta \dot{\bm{q}}_i(\bm{a},t)}=\rho_0 \dot{\bm{q}}_i,
\label{eq:fluid_md}
\end{equation}
and by a Legendre transform a generalized Hamiltonian quantity 
\begin{equation}
 H_f[\bm{q},\bm{\varpi}]=\int_D\left[\bm{\varpi} \dot{\bm{q}}-\mathcal{L}_f\right]\dif^3{\bm{a}}
=\int_D \left[\frac{\dot{\bm{\varpi}^2}}{2\rho_0}+E+\phi \right]\dif^3{\bm{a}},
\label{eq:fluid_Hamilotian}
\end{equation}
where $\rho_0\dot{\bm{q}}^2/2+E+\phi=\mathcal{H}_f$ can be considered as a Hamiltonian density.
The generalized Hamilton's equations of motion are 
\begin{equation}
 \dot{\varpi}_i=-\frac{\delta H_f}{\delta q_i}, \ \ \ \ \dot{q}_i=\frac{\delta H_f}{\delta \bm{\varpi}_i}.
\label{eq:fluid_heq}
\end{equation}
These equations can also be written in terms of the Poisson bracket defined by \cite{RevModPhys.70.467},
\begin{equation}
  \{F,G\}=\int_D \left[\frac{\delta F}{\delta \bm{q}}\cdot \frac{\delta G}{\delta \bm{\varpi}}-
            \frac{\delta G}{\delta q}\cdot \frac{\delta F}{\delta \bm{\varpi}}\right]\dif ^3 \bm{a}
\label{eq:fluid_psbr}
\end{equation}
viz.,
\begin{equation}
 \dot{\bm{\varpi}}_i=\{\bm{\varpi}_i,H_f\},\ \ \ \ \dot{\bm{q}}_i=\{\bm{q}_i,H_f\}
\end{equation}
Here $\delta q_i(\bm{a})/\delta q_j(\bm{a}')=\delta_{ij}\delta(\bm{a}-\bm{a}')$ has been used,  where $\delta(\bm{a}-\bm{a}')$ is a three-dimensional
 Dirac delta function.

If $a$ is fixed, Eq. (\ref{eq:fluid_heq}) represents the motion of a particle of the fluid with initial conditions 
$\left.q_i \right|_{t=0}=const,\left. \varpi_i \right|_{t=0} \in \mathrm{R}$, i.e. a set of phase curves. 
In contrast, Eq. (\ref{eq:ex2-6}) with initial conditions $\left.q_i \right|_{t=0}=const,\left. \dot{q} \right|_{t=0}=const$ represents 
a phase curve. There is a precedent, in that \cite{Morrison1980} proposed a Hamiltonian description of
 Poisson-Vlasov equations with a Hamiltonian quantity, which is an integral over phase space.
 These ideas of \cite{RickSalmon1988} and \cite{RevModPhys.70.467,Morrison1980} has given impetus to consider the mechanical 
system (\ref{eq:eq1}) as a special fluid constituting a collection of fluid particles in phase space. Therefore, 
we label a particle in the phase space by
\begin{equation}
 \bm{a}=(\bm{q}_0,\dot{\bm{q}}_0)=(q_0^1,\dots,q_0^n,\dot{q}_0^1,\dots,\dot{q}_0^n),
\label{eq:lag_vrb}
\end{equation}
the coordinate of a particle in the configuration space by
\begin{equation}
 \bm{q}=\bm{q}(\bm{a},t)=(q_1(\bm{a},t),\dots,q_n(\bm{a},t),
\end{equation}
and set $\rho_o=1$.
One thing to be noted is that, Eq. (\ref{eq:ex2-6}) describes the motion of individual particles without interaction and with constant density $\rho=1$. 
In continuum mechanics, the internal energy function $E$ is the source of the interaction (stress).
Therefore, analogous to Eq. (\ref{eq:fluid_lagdens}), one can consider $\hat{L}$ in Eq. (\ref{eq:ex2-7a}) as a Lagrangian 
density of system (\ref{eq:eq1})
\begin{equation}
\mathcal{L}=\hat{L}=\frac{1}{2}\dot{\bm{q}}^T\dot{\bm{q}}-\int_{\bm{0}}^{\bm{q}}(\tensor{\check{K}}\bm{q})^T\dif \bm{q}-
\int_{\bm{0}}^{\bm{q}} (\tilde{\tensor{K}}\bm{q})^T \dif \bm{q}.
\label{eq:Lag_denst}
\end{equation}
Thus, the Lagrangian functional of Eq. (\ref{eq:eq1}) can be presented as follows:
\begin{equation}
L[q,\dot{q}]=\int_D \mathcal{L}\dif^{2n} \bm{a}=\int_D\left[\frac{1}{2}\dot{\bm{q}}^T\dot{\bm{q}}-\int_{\bm{0}}^{\bm{q}}(\tensor{\check{K}}\bm{q})^T\dif \bm{q}-
\int_{\bm{0}}^{\bm{q}} (\tilde{\tensor{K}}\bm{q})^T \dif \bm{q} \right]\dif^{2n} \bm{a},
\label{eq:nh2}
\end{equation}
where $\dif^{2n}=\dif^n \bm{q}_0\dif^n \dot{\bm{q}}_0=\dif q^1_0\dots\dif q^n_0\dif \dot{q}^1_0\dots\dif\dot{q}^n_0$, 
from which the action functional can be presented as follows:
\begin{equation}
 S[q]=\int^{t1}_{t0}L[q,\dot{q}]\dif t
=\int^{t1}_{t0}\dif t\int_D \left[\frac{1}{2}\dot{\bm{q}}^T\dot{\bm{q}}-\int_{\bm{0}}^{\bm{q}}(\tensor{\check{K}}\bm{q})^T\dif \bm{q}-
\int_{\bm{0}}^{\bm{q}} (\tilde{\tensor{K}}\bm{q})^T \dif \bm{q} \right]\dif^{2n}\bm{a}
\label{eq:nh3}
\end{equation}
According to Hamilton's theorem, we have the functional derivative $\delta S/\delta \bm{q}(a,t)=0$:
\begin{eqnarray}
 \frac{\delta S}{\delta \bm{q}(\bm{a},t)}&=&\frac{\partial \mathcal{L}}{\partial \bm{q}(\bm{a},t)}
-\frac{\dif}{\dif t}\frac{\partial\mathcal{L}}{\partial \dot{\bm{q}}(\bm{a},t)} \nonumber\\
                               &=&-\ddot{\bm{q}}(\bm{a},t)-\tensor{\check{K}}\bm{q}-\tilde{\tensor{K}}\bm{q}=0 
\label{eq:nh4}
\end{eqnarray}
The equation above implies that under the initial conditions $\bm{a}$, a conservative system exists, the control equation of which
is Eq. (\ref{eq:ex2-6}), the phase curve of which coincides with that of a damped oscillator.
A canonical momentum density for the dissipative system (\ref{eq:eq1}) can be defined as the functional derivative
\begin{equation}
 \pi_i(\bm{a},t)=\frac{\delta L}{\delta \dot{q}_i(\bm{a})}=\dot{\bm{q}_i},
\label{eq:nh5}
\end{equation}
, while the classical canonical momentum is defined as a partial derivative.
By a Legendre transform, we have the generalized Hamiltonian $H_{Dp}$
\begin{equation}
 H_{Dp}[\bm{\pi},\bm{q}]=\int_D \dif^{2n}\bm{a}\left[\bm{\pi} \cdot \dot{\bm{q}}-\mathcal{L}\right]
=\int_D \dif^{2n}\bm{a} \left[\frac{1}{2}\bm{\pi}^T\bm{\pi}+\int_{\bm{0}}^{\bm{q}}(\tensor{\check{K}}\bm{q})^T\dif \bm{q}+
\int_{\bm{0}}^{\bm{q}} (\tilde{\tensor{K}}\bm{q})^T \dif \bm{q}\right],
\label{eq:nh6}
\end{equation}
where $\bm{q}=\bm{q}(\bm{a},t)$, with Hamiltonian density 
\[
 \mathcal{H}_{Dp}=\frac{1}{2}\dot{\bm{q}}^T\dot{\bm{q}}+\int_{\bm{0}}^{\bm{q}}(\tensor{\check{K}}\bm{q})^T\dif \bm{q}+
\int_{\bm{0}}^{\bm{q}} (\tilde{\tensor{K}}\bm{q})^T \dif \bm{q}.
\]
Thus, the generalized Hamilton's equations of motion for the dissipative system (\ref{eq:eq1}) are
\begin{equation}
 \dot{\pi}_i=-\frac{\delta H_{Dp}}{\delta q_i}, \ \ \dot{q}_i=\frac{\delta H_{Dp}}{\delta \pi_i} .
\label{eq:nh7}
\end{equation}
\begin{definition}
For two functionals $F[\bm{\pi}(\bm{a}),\bm{q}(a)]$ and $G[\bm{\pi}(a),\bm{q}(a)]$, there exists in a domain $D$ of the phase space, a functional 
called the Poisson bracket of $F$ and $G$
\begin{equation}
 \{F,G\}[\bm{\pi}(\bm{a}),\bm{q}(\bm{a})]=\int_D \left[\frac{\delta F}{\delta \bm{q}(\bm{a}')}\cdot \frac{\delta G}{\delta \bm{\pi}(\bm{a}')}-
            \frac{\delta G}{\delta \bm{q}(\bm{a}')}\cdot \frac{\delta F}{\delta \bm{\pi}(\bm{a}')}\right]\dif ^{2n} \bm{a},
\label{eq:pos_def}
\end{equation}
\end{definition}
The Generalized Hamilton's equations of motion can also be represented in terms of the Poisson bracket (\ref{eq:pos_def})
viz.,
\begin{equation}
 \dot{\pi}_i=\{\pi_i,H_{Dp}\}
,\dot{q}_i=\{q_i,H_{Dp}\}.
\label{eq:nh9}
\end{equation}

Expanding $\{\pi_i,H_{Dp}\}$, we have 
\begin{eqnarray}
 \{\pi_i(\bm{a}),H_{Dp}\}&=&\frac{\delta \pi_i(\bm{a})}{\delta q_j(\bm{a}')}\frac{\delta H_{Dp}}{\delta \pi_j(\bm{a}')}-
                \frac{\delta \pi_i(\bm{a})}{\delta \pi_j(\bm{a}')}\frac{\delta H_{Dp}}{\delta q_j(\bm{a}')}\nonumber\\
                   &=&-\delta_{ij}\delta(\bm{a}-\bm{a}')\frac{\delta H_{Dp}}{\delta q_j(\bm{a}')}\nonumber\\
                   &=&-\frac{\delta H_{Dp}}{\delta q_i(\bm{a})},
\label{eq:expand_1}
\end{eqnarray}
Here $\delta \pi_i(\bm{a})/\delta \pi_j(\bm{a}')=\delta_{ij}\delta(\bm{a}-\bm{a}')$ has been used,  where $\delta(\bm{a}-\bm{a}')$ is a three-dimensional
 Dirac delta function.
Analogous to Eq. (\ref{eq:int2sum}),
\begin{equation}
 \int_D \dif^{2n} \bm{a}\leftrightarrow \sum_i, \ \ H_{Dp}=\sum_i \mathcal{H}_{Dp}(\bm{a}).
\label{eq:int2sum1}
\end{equation}
According to Eq. (\ref{eq:int2sum1}), we can derive from Eq. (\ref{eq:expand_1})
\begin{equation}
\dot{\pi}_i(\bm{a})=\{\pi_i(\bm{a}),H_{Dp}\}=-\frac{\delta H_{Dp}}{\delta q_i(\bm{a})}=-\frac{\partial \mathcal{H}_{Dp}(\bm{a})}{\partial q_i(\bm{a})},
\label{eq:inf2fin_1}
\end{equation}
and similarly, 
\begin{equation}
\dot{q}_i(\bm{a})=\{q_i(\bm{a}),H_{Dp}\}=\frac{\delta H_{Dp}}{\delta \pi_i(\bm{a})}=\frac{\partial \mathcal{H}_{Dp}(\bm{a})}{\partial \pi_i(\bm{a})}
\label{eq:inf2fin_2}
\end{equation}
Therefore, we can assert that Eq. (\ref{eq:inf2fin_1}) and Eq. (\ref{eq:inf2fin_2}) describe a phase curve which is a common phase 
curve of the dissipative system and a conservative system subject to the initial conditions $\bm{a}$.

From the Hamilton's equation of motion (\ref{eq:nh7}), we can derive the total energy conservative principle
\begin{eqnarray*}
 \delta H_{Dp}&=&\int_D\left[\frac{\delta H_{Dp}}{\delta q_i(\bm{a})}\delta q_i(\bm{a})
          +\frac{\delta H_{Dp}}{\delta \pi_i(\bm{a})}\delta \pi_i(\bm{a})\right]\dif ^{2n}\bm{a}\\
         &=&\int_D\left[\frac{\delta H_{Dp}}{\delta q_i(\bm{a})}\frac{\dif q_i(\bm{a})}{\dif t}\dif t+ 
          \frac{\delta H_{Dp}}{\delta \pi_i(\bm{a})}\frac{\dif \pi_i(\bm{a})}{\dif t}\dif t \right]\dif ^{2n}\bm{a}\\
         &=&\int_D\left[\frac{\delta H_{Dp}}{\delta q_i(\bm{a})}\frac{\delta H_{Dp}}{\delta \pi_i(\bm{a})}\dif t-
             \frac{\delta H_{Dp}}{\delta \pi_i(\bm{a})}\frac{\delta H_{Dp}}{\delta q_i(\bm{a})}\dif t \right]\dif ^{2n}\bm{a}\\
         &=&0
\end{eqnarray*}

\section{Conclusion}
The following conclusions can be drawn. The infinite-dimensional description(\ref{eq:nh5}),(\ref{eq:nh6}),
(\ref{eq:nh7}),(\ref{eq:pos_def},(\ref{eq:nh9}) can describe a dissipative mechanical system based on Proposition \ref{pro:1}: 
For any non-conservative classical mechanical system and arbitrary initial condition, there exists a conservative system; both systems sharing
one and only one common phase curve; and the value of the Hamiltonian of the conservative system is equal to the sum of the total energy of 
the non-conservative system on the aforementioned phase curve
and a constant depending on the initial condition. In fact, if the generalized Hamilton's equation of motion in (\ref{eq:nh7}), (\ref{eq:nh9}) is 
constrained subject to initial conditions $\bm{a}$, the generalized Hamilton's equations determine a phase curve of the afore-mentioned conservative 
system (\ref{eq:ex2-6}). As the classical Hamilton's equations represent the conservation of mechanical energy principle, 
the generalized Hamilton's equations of (\ref{eq:nh7},\ref{eq:nh9}) describe an analogous the conservation of total energy principle. One can 
assert that the generalized Hamilton's equations (\ref{eq:nh7}),(\ref{eq:nh9}) are the generalization of the classic Hamilton's equations.

\end{document}